\documentclass{article}
\usepackage{spconf}
\usepackage{amsfonts}                                           
\usepackage{amssymb}                                           
\usepackage{bbold}
\usepackage{amsmath,mathrsfs}						                   
\usepackage{amsxtra}                                            
\usepackage{amsthm}						                      
\usepackage{MnSymbol}
\usepackage{accents}
\usepackage{times}
\usepackage[english]{babel}
\usepackage{booktabs}                                               
\usepackage{textcomp}                                               
\usepackage[dvips,pdfborder={0 0 0},hypertexnames=false,extension=pdf]{hyperref}
\usepackage{epsfig}
\usepackage[utf8]{inputenc}                                         
\usepackage[T1]{fontenc}
\usepackage{ifpdf}
\ifpdf
  \usepackage[pdftex]{graphicx}
  \DeclareGraphicsExtensions{.pdf}
\else
  \usepackage{graphicx}
  \usepackage{epsfig}
  \DeclareGraphicsExtensions{.eps}
\fi

\usepackage[style=numeric,
  firstinits=true, 
            maxbibnames=99,
            isbn=false,
            natbib=true,
            hyperref=true,
            bibencoding=utf8]{biblatex}
\AtEveryBibitem{\clearfield{note}}   
\DefineBibliographyStrings{english}{%
  references = {},
}

\bibliography{jabref_philipp_utf2.bib}

\usepackage{caption}
\usepackage{subcaption}
\usepackage[all]{xy}
\usepackage{xytree}
\usepackage{multirow}

\usepackage{rotating}   
\usepackage{verbatim}
\usepackage{color}
\usepackage{srcltx}                                                 
\usepackage{fancyhdr}                                              
\usepackage{enumerate}                                              
\graphicspath{{Eps/}}

 \hypersetup {
  pdftitle          = {},
  pdfauthor 	    = {Philipp Walk and Peter Jung},
  pdfsubject 	    = {},
  pdfkeywords 	    = {},
  bookmarksnumbered = true,   
}

\newcommand{\alp}{\ensuremath{\alpha}}

\newcommand{\bet}{\ensuremath{\beta}}

\newcommand{\del}{\ensuremath{\delta}}

\newcommand{\tht}{\ensuremath{\theta}}

\newcommand{\ti}{\ensuremath{\tilde{i}}}
\newcommand{\tI}{\ensuremath{\tilde{I}}}
\newcommand{\tJ}{\ensuremath{\tilde{J}}}
\newcommand{\tj}{\ensuremath{\tilde{j}}}

\newcommand{\gam}{\ensuremath{\gamma}}

\newcommand{\lam}{\ensuremath{\lambda}}

\newcommand{\ome}{\ensuremath{\omega}}

\newcommand{\hx}{\ensuremath{u}}

\newcommand{\tn}{\ensuremath{\tilde{n}}}

\newcommand{\foral}{\ensuremath{\quad{\text{for all}\quad}}}

\newcommand{\Omi}{{\ensuremath{\mathcal{O}}}}

\newcommand{\G}{\ensuremath{\mathcal{G} }}						

\newcommand{\C}{{\ensuremath{\mathbb C}}}

\newcommand{\N}{{\ensuremath{\mathbb N}}}

\newcommand{\Z}{{\ensuremath{\mathbb Z}}}

\newcommand{\RA}{\ensuremath{\Rightarrow} }

\newcommand{\ra}{\rightarrow}

\newcommand{\LRA}{\ensuremath{\Leftrightarrow} }

\newcommand{\vx}{{\ensuremath{\mathbf x}}}

\newcommand{\hvx}{{\ensuremath{\mathbf u}}}

\newcommand{\vy}{\ensuremath{\mathbf y}}

\newcommand{\by}{\ensuremath{\tilde{y}}}

\newcommand{\cby}{\ensuremath{\cc{\tilde{y}}}}
\newcommand{\hy}{{\ensuremath{v}}}
\newcommand{\hvy}{{\ensuremath{{\mathbf v}}}}

\newcommand{\va}{{\ensuremath{\mathbf a}}}

\newcommand{\vby}{{\ensuremath{\tilde{\mathbf y}}}}

\newcommand{\zero}{{\ensuremath{\mathbf 0}}}

\newcommand{\tx}{\ensuremath{\tilde{x}}}

\newcommand{\ty}{\ensuremath{\tilde{y}}}

\newcommand{\vtx}{\ensuremath{{\tilde{\mathbf x}}}}
\newcommand{\vty}{\ensuremath{{\tilde{\mathbf y}}}}

\newcommand{\vB}{{\ensuremath{\mathbf B}}}

\newcommand{\vu}{{\ensuremath{\mathbf u}}}

\newcommand{\vv}{{\ensuremath{\mathbf v}}}

\newcommand{\vP}{\ensuremath{\mathbf P}}                         
\newcommand{\vPA}{\ensuremath{\mathbf P}_A}                         
\newcommand\diam{\operatorname{diam}} 

\newcommand{\cred}[1]{{\color{red}{\ensuremath{#1}}}}
\newcommand{\cgre}[1]{{\color{green}{\ensuremath{#1}}}}

\newcommand{\noi}{\noindent}

\ifx\definition\undefined
\newtheorem{definition}{Definition}         
\fi
\ifx \@definition \@empty
\newtheorem{definition}[theorem]{Definition}   
\fi

\ifx\conjecture\undefined
\newtheorem{conjecture}{Conjecture}         
\fi
\ifx\theorem\undefined
\newtheorem{theorem}{Theorem}         
\fi
\ifx\lemma\undefined
\newtheorem{lemma}{Lemma}
\fi
\ifx\question\undefined
\newtheorem{question}{\color{red}{Question}}         
\fi
\ifx\proposition\undefined
\newtheorem{proposition}[theorem]{Proposition} 
\fi
{
\comment\noi}%
{\endcomment}

{\par\noindent{\em Beweis\/}.}%
{\hspace*{\fill}{\qed}\vspace{1ex}\par}
{\par\noindent{\em Proof\/}.}%
{\par}

{\hspace*{\fill}{}\vspace{1ex}\par}
{\par\vspace{1.5ex}\noindent{\em Remark\/}.}
{\par\vspace{1.5ex}}
\ifx\remark\undefined
\newenvironment{remark}%
{\par\vspace{1.5ex}\noindent{\em Remark\/}.}
{\par\vspace{1.5ex}}
{\par\vspace{1.5ex}\noindent{\em Example\/}. }
{\par\vspace{1.5ex}}
{\noi\vspace{0.5ex}\small}
{\vspace{0.5ex}\par\normalsize}

\newcounter{Examplecount}
{\color{gray}}
{\vspace{0.5ex}\par\normalsize}

{\renewcommand{\labelenumi}{(\roman{enumi})}\begin{list}{\labelenumi}
{\usecounter{enumi}\setlength{\labelwidth}{1.5cm}\setlength{\topsep}{0.3cm}\setlength{\itemsep}{-3pt}}}
{\end{list}}
{\renewcommand{\labelenumi}{(\arabic{enumi})}\begin{list}{\labelenumi}
{\usecounter{enumi}\setlength{\labelwidth}{1.5cm}\setlength{\topsep}{0.3cm}\setlength{\itemsep}{-3pt}}}
{\end{list}}
{\renewcommand{\labelenumi}{$\bullet$}\begin{list}{\labelenumi}
{\setlength{\labelwidth}{1.5cm}\setlength{\topsep}{0.3cm}\setlength{\itemsep}{-2pt}}}
{\end{list}}

\makeatletter
\newcommand{\set}[2]{\ensuremath{%
\setbox0=\hbox{\ensuremath{#2}}
\dimen@\ht0
\advance\dimen@ by \dp0
\left\{\left.#1\rule[-\dp0]{0pt}{\dimen@}\;\right|\;#2\right\} }}
\makeatother

\newcommand{\Norm}[1]{\ensuremath{ \left\|#1\right\| }}

\newcommand{\skprod}[1]{\ensuremath{ \left\langle #1 \right\rangle }}

\newcommand{\cc}[1]{{\ensuremath{\overline{#1}}}} 

\DeclareMathOperator{\supp}{supp}

\newcommand{\namen}[1]{{\textsc{#1}}}           

\ifx \@paragraph \@empty
\makeatletter
\renewcommand\paragraph{\@startsection
{paragraph}{4}{\z@}{-3.5ex plus-1ex minus-.2ex}%
{1.3ex plus.2ex}{\normalfont\itshape}}

\fi

\DeclareUnicodeCharacter{04B4}{\CYRTETSE}
\definecolor{gray}{rgb}{0.3,0.3,0.3}
{\color{black}}                      
{\color{black}}

\newtheorem*{theoremno}{Theorem} 
\newtheorem*{lemmano}{Lemma} 
\begin{document}

\onecolumn
\renewcommand{\tn}{n}

\title{A Stability Result for Sparse Convolutions}
\twoauthors{Philipp Walk}{
  Institute of Theoretical Information\\
  Technology, Technical University\\
  of Munich (TUM)\\
  \emph{\small philipp.walk@tum.de}}
{Peter Jung}{
  Heinrich-Hertz-Chair for Information Theory\\
  and Theoretical Information Technology\\
  Technical University Berlin (TUB)\\
  \emph{\small peter.jung@mk.tu-berlin.de}
}

\maketitle

\begin{abstract}
We will establish in this note a stability result for sparse convolutions on torsion-free additive (discrete) abelian groups. 
Sparse convolutions on torsion-free groups are free of cancellations and hence admit stability,
i.e. injectivity with a universal lower bound $\alpha=\alpha(s,f)$, only depending on the cardinality 
$s$ and $f$ of the supports of both input sequences. More precisely, we show that $\alpha$ depends only
on $s$ and $f$ and not on the ambient dimension. This statement
follows from a reduction argument which involves a compression into a small set preserving 
the additive structure of the supports.
\end{abstract}

\section{Introduction}

In this work, we will prove a $\ell^2-$norm inequality for $(s,f)-$sparse convolutions on $\ell^2_s (G)
\times \ell^2_f (G)$ for one--dimensional abelian torsion-free discrete groups $G=(G,+,|\cdot|)$ equipped with the
counting measure $|\cdot|$.
We define for a natural number $k$ the set of \emph{$k$--sparse sequences} in $\ell^2(G)$:
\begin{align}
  \ell^2_k (G):=\set{\vx:G\to \C}{\Norm{\vx}^2:=\sum_{i\in G} |x_i|^2 < \infty, |\supp \vx | \leq k}.
\end{align}
Then for two  $(s,f)-$sparse sequences $\vx\in \ell^2_s (G)$ and $\vy\in\ell^2_f (G)$ its \emph{convolution} $\vx * \vy$ 
is given by the sequence with elements:
\begin{align}
  (\vx * \vy)_{j}= \sum_{i\in G} x_i y_{j-i} \foral j \in G \label{eq:defconv}
\end{align}
each being a finite sum.  Let us define the set of \emph{$k-$sparse vectors} in $\C^n$ by $\Sigma^n_k$ and the set of all
support sets with cardinality $k$ by $[0,n-1]_k:=\set{A\subset \{0,1,\dots,n-1\}}{|A|=k}$. 
For any $A\in [0,n-1]_k$, the
projection operator $\vPA :\C^n \to \C^k$ cuts out from the $n\times n-$matrix $\vB$ an $k\times k$ principal submatrix
$\vB^A=\vPA \vB \vPA^*$.  Further, we denote by $\vB_\va$ an $n\times n-$ \emph{Hermitian Toeplitz matrix}  generated by
the autocorrelation $b_k(\va)=\sum_{l} \cc{a_l} {a_{l+k}}$ of $\va\in \Sigma^n_k$ with \emph{symbol} 
\begin{align}
  \begin{split}
   b (\va,\ome)=  \sum_{k=-(\tn-1)}^{\tn-1} b_k(\va) e^{\imath k\ome}
       = 1 + \sum_{k=1}^{\tn-1}\left( \mu_k \cos(k\ome) + \nu_k \sin(k\ome)\right)\\
       \text{with}\quad \mu_k:=2\Re(b_k(\va))\quad \text{and}\quad \nu_k:=-2\Im(b_k(\va)),
    \end{split} \label{eq:bsymbol3}\quad,\quad\ome\in[0,2\pi) 
 \end{align}
 which defines a \emph{trigonometric polynomial of order not larger than $n$}. 
 %

\section{Main Result and Proof}\label{sec:mainresult}
%
The following theorem is a generalization of a result in \cite{WJ13a}, (i) in the sense of the extension to infinite
sequences on torsion-free abelian groups (note if one adds consecutive $n-1$ zeros to  sparse vectors in $\Sigma_k^n$ the circular
convolution in $\C^{2n-1}$ can be written  with $G=(\Z/(2n-1)\Z],\oplus)$ by \eqref{eq:defconv},  since the addition modulo
$\oplus$ (modulo $2n-1$) equals then the regular addition $+$
(ii) extension to the complex case, which actually only replaces \namen{Szeg{\"o}} factorization with \namen{Fejer-Riesz} factorization
in the proof and (iii) with a precise determination of the dimension parameter $\tn$\footnote{Actually, the estimate of
the dimension $\tn=\tilde{n}$ of the constant
$\alp_{\tilde{n}}$ in \cite{WJ12b}, was quite too optimistic.}.
%
\begin{theoremno}\label{thm:dryi}
   Let $s$ and $f$ be natural numbers and $G$ a one--dimensional torsion-free, discrete, additive abelian group. Then there exist
   constants $0<\alp(s,f)\leq\beta(s,f)=\sqrt{\min\{s,f\}}<\infty$ depending solely on $s$ and $f$,
   s.t. for all $\vx\in\ell^2_s (G)$ and $\vy\in \ell^2_f (G)$ 
   \begin{align}
      \alp(s,f) \Norm{\vx}\Norm{\vy} \leq \Norm{\vx * \vy} \leq \beta(s,f)\Norm{\vx}\Norm{\vy} \label{eq:dryi}
   \end{align}
  holds.
  Moreover, we have with $n=\lfloor 2^{2(s+f-2)\log_2 (s+f-2)}\rfloor+1:$
  \begin{align}
    \alp^2 (s,f)=\min\big\{
    \min_{\substack{\vty\in \Sigma^{n}_f,\Norm{\vty}=1\\ I\in [0,n-1]_s}}
    \lam(\vB^I_\vty),
    \min_{\substack{\vtx\in \Sigma^{n}_s,\Norm{\vtx}=1\\ J\in [0,n-1]_f}}
    \lam(\vB^J_\vtx)\big\},
  \end{align}
  which is a decreasing sequence in $s$ and $f$. For $\beta(s,f)=1$ we get equality with $\alp(s,f)=1$. 
\end{theoremno}
%
\begin{proof}
  The upper bound is trivial and follows from the \namen{Cauchy-Schwartz} inequality and the \namen{Young} inequality
  for $p=1,q=r=2$. For $\vx=0$ or $\vy=0$ the inequality is trivial as well, hence we assume that $\vx$ and $\vy$ are
  non-zero. If $|\supp \vx|=1$ then there exist $i\in G$ sucht that $x_i\not=0$ and $x_j=0$ for all $i\not=j\in \G$.
  The norm of the convolution equals then $\Norm{\vx*\vy}=\Norm{x_i\vy}$ and the inequality \eqref{eq:dryi} becomes an equality.\\
  We consider therefore the normalized version of the convolution for $s,f\geq 2$, i.e. the following problem: 
  \begin{equation}
    \alp(s,f):=\inf_{\substack{\zero\not=\vx\in \ell^2_s (G)\\ \zero\not=\vy\in \ell^2_f (G)}} \frac{\Norm{\vx * \vy}}{\Norm{\vx}\Norm{\vy}}=
    \inf_{\substack{\vx\in \ell^2_s (G), \vy\in \ell^2_f (G)\\ \Norm{\vx}=\Norm{\vy}=1}} \Norm{\vx * \vy}\label{eq:pb}.
  \end{equation}
  This is a \emph{bi-quadratic optimization problem} which is known to be NP-hard in the general case \cite{LNQY09}.
  The squared norm of the convolution of two finitely supported sequences is given by \eqref{eq:defconv} as:
  \begin{align}
   \Norm{\vx*\vy}^2=\sum_{j\in G} \left|\sum_{i\in G} x_{i} y_{j-i}\right|^2.\label{eq:normsquarestart}
  \end{align}
  Let $I$ and $J$ be sets of $G$ such that
  $\supp\vx\subseteq I$ and $\supp\vx\subseteq J$ with $|I|= s,|J|= f$ for some
  $2\leq s,f\in \N$.  For such $I,J\subset G$ with
  $I=\{i_0,\dots,i_{s-1}\}$ and $J=\{j_0,\dots,j_{f-1}\}$ (ordered sets) we can represent $\vx$ and $\vy$ by 
  complex vectors $\hvx \in \C^s$ and $\hvy \in \C^f$ component-wise given by:
  \begin{align}
   x_i = \sum_{\tht=0}^{s-1} \hx_\tht \del_{i,i_\tht}\quad,\quad  y_j = \sum_{\gam=0}^{f-1} \hy_\gam \del_{j,j_\gam}
   \foral i,j\in G.
  \end{align}
  Inserting this representation in \eqref{eq:normsquarestart} yields: 
  \begin{align}
    \Norm{\vx*\vy}^2&=\sum_{j\in G}\Bigg|\sum_{i\in G}\left( \sum_{\tht=0}^{s-1} \hx_{\tht}\del_{i,i_{\tht}}\right)
    \left( \sum_{\gam=0}^{f-1} \hy_\gam \del_{j-i,j_\gam}\right)\Bigg|^2\\
    &= \sum_{j\in G} \Bigg| \sum_{\tht=0}^{s-1} \sum_{\gam=0}^{f-1} \left( \sum_{i\in G} \hx_{\tht}\del_{i,i_{\tht}} 
     \hy_\gam \del_{j,j_\gam +i}\right)\Bigg|^2.\label{eq:ijsum}
  \intertext{Since the inner $i-$sum is over  $G$, we can shift $I$ by $i_0$ if we set $i\ra i+i_0$ (note that $\vx\not=0$), 
     without changing the value of the sum: }
   &= \sum_{j\in G} \Bigg| \sum_\tht \sum_\gam 
   \left( \sum_{i\in G} \hx_{\tht} \del_{i+{i_0},i_{\tht}} \hy_\gam \del_{j,j_\gam +i + i_0}\right)\Bigg|^2.
   \intertext{By the same argument we can shift $J$ by setting $j\ra j +i_0 +j_0$ and get:}
   &= \sum_{j\in G} \Bigg| \sum_\tht \sum_\gam 
   \left( \sum_{i\in G} \hx_{\tht} \del_{i,i_{\tht}-i_0} \hy_\gam \del_{j,j_\gam -j_0 + i}\right)\Bigg|^2.
  \intertext{Therefore, we always can assume that the supports $I,J\subset G$ have $i_0=j_0=0$ in $\eqref{eq:pb}$. From \eqref{eq:ijsum} we get:}
   &= \sum_{j\in G} \Bigg| \sum_\tht \sum_\gam \left( \hx_{\tht} \hy_\gam \del_{j,j_\gam +{i_\tht}} \right)\Bigg|^2 \\
   &= \sum_{j\in G} \sum_{\tht,\tht'}\sum_{\gam,\gam'}  \hx_{\tht} \cc{\hx_{\tht'}} \hy_\gam \cc{\hy_{\gam'}} 
   \del_{j,j_{\gam} +i_{\tht}} \del_{j,j_{\gam'} +i_{\tht'}}=:b_{I,J}(\vu,\vv)\label{eq:bijhxhy}.
  \end{align}
Usually, fourth order tensors like $\del_{i_\tht+j_{\gam},i_{\tht'} +j_{\gam'}}$ make such bi-quadratic optimization
problems over $\C^s \times \C^f$ NP-hard, see \cite{LNQY09}.  The interesting question is now: what is the smallest
dimension to represent this tensor, i.e. preserving the additive structure?  Let us consider a mapping $\phi$ of the
indices.  For $I,J\subset G$ with $0\in I\cap J$ an injective map:
\begin{align}
 \phi:I +J \to \Z
\end{align}
which additional satisfies (preserves additive structure of the indices):
\begin{align}
  \forall i,i'\in I, j,j'\in J : i+j=i' + j' \RA \phi(i) + \phi(j) = \phi(i') + \phi(j') \label{eq:f2isodef}
\end{align}
is called a \emph{Freiman homomorphism on $I,J$} and is a \emph{Freiman isomorphism} if:
\begin{align}
  \forall i,i'\in I, j,j'\in J : i+j=i' + j' \LRA \phi(i) + \phi(j) = \phi(i') + \phi(j') \label{eq:f2isodef2},
\end{align}
see e.g. \cite[pp.299]{Gry13}.  If we could show that $\phi(I),\phi(J)\subset [0, n-1]=\{0,1,\dots,n-1\}$, where
$n=n(s,f)$, for any $I,J\subset G$ with $|I|=s,|J|=f$ the minimization problem reduces to an $n$--dimensional problem.
Indeed, this was a conjecture by \namen{Konyagin} and \namen{Lev} \cite{KL00}, which was proved very recently by
\namen{Grynkiewicz} in \cite[Theorem 20.10]{Gry13} for Freiman dimension $d=1$. He could even prove a  more generalized
compression argument of arbitrary sum sets with finite Freiman dimension $d$ in torsion-free abelian groups. We will
state here a restricted version of his result for additive abelian groups with  two sets $A_1$ and $A_2$:
\begin{lemmano}[\cite{Gry13}]\label{thm:gry2}
  Let $G$ be a  torsion-free additive abelian group and  $A_1,A_2\subset G$ be finite sets containing zero with
  $m:=|A_1\cup A_2|$ and having finite Freiman dimension $d=\dim^+(A_1+A_2)$. Then there exists an injective Freiman
  homomorphism: 
  \begin{align}
     \phi: A_1 + A_2 \to \Z
  \end{align}
  such that 
  \begin{align}
    \diam(\phi(A_1)),\diam(\phi(A_2))\leq d!^2\left( \frac{3}{2}\right)^{d-1}\cdot 2^{m-2} + \frac{3^{d-1}-1}{2}.
  \end{align}
\end{lemmano}
For simplicity we have restricted our statements here solely on discrete (countable) groups.  Thus, for
$A_1=A_2=A:=I\cup J$ we get $|A\cup A|=|A|\leq m$. By a simple upper bound for $d$ in \cite[Corollary 5.42]{Tao06} we
get $d\leq m-2$ for $m\geq $ and by Grynkewiecz a bijective Freiman homomorphism on $A+A$, which is a Freiman isomorphism on $A$ with 
\begin{align}
  \diam(\phi(A)) \leq (m-2)!^2 \left(\frac{3}{2}\right)^{m-3}2^{m-2} + \frac{3^{m-3} -1}{2} < \lfloor 2^{2(m-1)\log_2(m-1)}\rfloor.
\end{align}
Since $|I|=s, |J|=f$ with $s,f\geq 2$ and $0\in I\cap J$ we always have $3\leq m\leq s+f-1$. Then $\diam(\phi(A)):=\max(\phi(A))
-\min(\phi(A))$ and hence $\diam(\phi(I)\cup \phi(J))\leq n$ with
\begin{align}
  n:=\lfloor 2^{2(s+f-2)\log_2 (s+f-2)} \rfloor +1 \label{eq:n}.
\end{align}
Let us now define the minimum in the image:
\begin{align}
  c^*:=\min_{c\in I\cup J} \phi(c).
\end{align}
Then we can translate the Freiman isomorphism by setting $\phi'=\phi-c^*$ (still satisfy \eqref{eq:f2isodef}) and define 
$\tI:=\phi'(I)$ and $\tJ:=\phi'(J)$. With $n$ in \eqref{eq:n} we have:
\begin{align}
  0\in\tI\cup \tJ\subset\{0,1,2,\dots ,n-1\}=[0,\tn-1].
\end{align}
Unfortunately, a Freiman isomorphism does not necessarily preserve the order. 
However, this is not a problem, since we only need is to know, that indices not larger than $\tn-1$ are needed to express
the combinatorics of the convolution, i.e. \eqref{eq:bijhxhy} reads now as:
\begin{align}
  b_{I,J} (\hvx,\hvy)
  &=\sum_{\tht,\tht'} \sum_{\gam,\gam'} 
  \hx_\tht \cc{\hx_{\tht'}} \hy_{\gam} \cc{\hy_{\gam'}} \del_{\ti_{\tht} +\tj_\gam,\tj_{\gam'} +\ti_{\tht'}}\label{eq:pipjhxhy}.
\end{align}
and the norm of the convolution is indeed reduced to $n$ dimensions. Next, we can define the embedding of $\hvx,\hvy$ into $\C^{n}$ by:
\begin{align}
  \tx_i = \sum_{\tht=0}^{s-1} \hx_{\tht} \del_{i,{\ti_{\tht}}} \ ,\ 
  \ty_j = \sum_{\gam=0}^{f-1} \hy_{\gam} \del_{j,{\tj_{\gam}}} \foral i,j \in [0,\tn-1].\label{eq:xytn}
\end{align}
Since for all $\tht$ and $\gam$ there exist unique $\ti_\tht\in \tI$ resp. $\tj_\gam\in \tJ$ ($\phi'$ is bijective) we get: 
 \begin{align}
    \hx_\tht = \sum_{i=0}^{\tn-1} \tx_{i} \del_{i,\ti_\tht} \quad,\quad\hy_\gam = \sum_{j=0}^{\tn-1} \ty_{j} \del_{j,\tj_\gam},
 \end{align}
 and inserting this into \eqref{eq:pipjhxhy} yields:
 \begin{align}
    b_{I,J} (\hvx,\hvy)
    &=\sum_{\tht,\tht'} \sum_{i,i'=0}^{\tn-1}  \hx_{\tht}\cc{\hx_{\tht'}} \del_{i,\ti_\tht}\del_{i',\ti_{\tht'}} 
    \sum_{\gam,\gam'} \sum_{j,j'=0}^{\tn-1}  \hy_{{\gam}} \cc{\hy_{\gam'}}
    \del_{j,\tj_{\gam}}\del_{j',\tj_{\gam'}} \del_{\tj_{\gam} +(\ti_{\tht}-\ti_{\tht'}),\tj_{\gam'}}\\
    \text{\eqref{eq:xytn}}\ra &=\sum_{i,i'=0}^{\tn-1} \tx_{i}\cc{\tx_{i'}}  \sum_{j,j'=0}^{\tn-1}  \ty_{j} \cc{\ty_{j'}} \del_{j+(i-i'),j'}\\
    &=\sum_{i,i'=0}^{\tn-1}  \tx_{i}\cc{\tx_{i'}}  \underbrace{\sum_{j=0}^{\tn-1}  \ty_{j} 
      \cc{\ty_{j +(i-i')}}}_{=:\cc{(\vB_\vty)_{i',i}}}=\skprod{\vtx,\vB_\vty \vtx}\label{eq:skprodvxvy},
 \end{align}
 where $\vB_\vty$ is a $n\times n$ Hermitian Toeplitz matrix with first row $(\vB_\vty)_{0,k}=\sum_{j=0}^{n-k}
 \cc{\ty_{j}} {\ty_{j+k}}=:b_k(\vty)$ resp. first column $(\vB_\vty)_{k,0}=:b_{-k}$ for $k\in [0,n-1]$ and \emph{symbol}
 $b(\vty,\ome)$ given by \eqref{eq:bsymbol3}, see e.g. \cite{BG05a}. We call $b(\vty,\ome)$ for each $\vty\in \C^{\tn}$
 with $\nu_0=\Norm{\vty}=1$ a \emph{normalized  trigonometric polynomial of order $\tn-1$}. Minimizing the scalar product
 in \eqref{eq:skprodvxvy} over all $\vtx\in \C^{\tn}$ with $\Norm{\vtx}=1$ defines the smallest eigenvalue of
 $\vB_{\vty}$:
 \begin{align}
    \lam(\vB_\vty):=\min_{\vtx\in \C^{\tn}, \Norm{\vtx}=1} \skprod{\vtx,\vB_\vty \vtx}.
 \end{align}
 By the well-known \namen{Fejer-Riesz} Factorization, see e.g. \cite[Thm.3]{Dim04}, we know that the symbol of
 $\vB_{\vby}$  is \emph{non-negative}
 \footnote{Note, there exist $\vby\in\C^{n}$ with $\Norm{\vby}=1$ and $b(\vby,\ome)=0$ for some
   $\ome\in[0,2\pi)$. Thats the reason why things are more complicated here. Moreover, we want to find a universal lower
   bound over all $\vby$, which is equivalent to a universal lower bound over all non-negative trigonometric polynomials of
   order $n-1$. By the best knowledge of the authors, there exist no analytic lower bound for $\alp(s,f)$.}
 for every $\vby\in\C^{n}$, i.e.  $0\leq\min_{\ome}b(\vby,\ome)$. By \cite[(10.2)]{BG05a}
 we then have $\lam(\vB_{\vby})>0$. Hence $\vB_{\vby}$ is invertible and the \emph{determinant}
 $\det(\vB_{\vby})\not=0$.  Using:
 \begin{align}
    \frac{1}{\lam(\vB_{\vby})} = \Norm{\vB^{-1}_{\vby}}
 \end{align}
 in \cite[p.59]{BG05a}, we can estimate the smallest eigenvalue (singular value) with \cite[Thm. 4.2]{BG05a} by the
 determinant as:
 \begin{align}
    \lam(\vB_{\vby})\geq |\det(\vB_{\vby})| \frac{1}{\sqrt{\tn} (\sum_k |b_k(\vby)|^2)^{(n-1)/2}}.
    \label{eq:eigdetbound}
 \end{align}
 In the following we will not further explicitely account for the sparsity of $\vby$ which may improve
 the next steps. For our purpose it will be sufficient here to show a non-zero lower bound.
 The $\ell^2$--norm of the sequence $b_k(\vby)$ can be upper bounded for $n>1$ by the \namen{Cauchy-Schwartz}
 inequality (instead one may also utilize the upper bound of the theorem):
 \begin{align}
    \sum_k |b_k(\vby)|^2 &\leq 1 + 2\sum_{k=1}^{n-1}|\sum_{j=0}^{n-1} \by_j \cby_{j+k}|^2
    \leq 1 +2 \sum_{k=1}^{n-1} \Norm{\vby}^4 = 1+2(\tn-1)< 2n,
 \end{align}
 which is independent of $\vby\in \C^{\tn}$ with $\Norm{\vty}=1$!  Since the determinant is a continuous function in $\vby$ over a
 compact set, the minimum is attained and is denoted by $0<d_{n}:=\min_{\vby} |\det(\vB_{\vby})|$. Note, that
 $d_{n}$ is a decreasing sequence, since we extend the minimum to a larger set by increasing $\tn$. Hence we get:
 \begin{align}
    \min_{\vby\in \C^{\tn},\Norm{\vty}=1}\left(
      |\det(\vB_{\vby})|\frac{1}{\sqrt{n}(2 n)^{(n-1)/2}}\right)=\frac{\sqrt{2}}{(2 n)^{\tn/2}} d_{n}.
 \end{align}
 This is a valid lower bound by \eqref{eq:eigdetbound} for the smallest eigenvalue of all  $\vB_{\vby}$.
 Hence we have shown
 \begin{align}
    \min_{\vty\in\C^{\tn},\Norm{\vty}=1} \lam(\vB_\vty) >\sqrt{2} {(2 n)}^{-\frac{n}{2}} d_{\tn}>0.\label{eq:Pquadratic}
 \end{align}
 Now, bringing the support back into play, we see that $\vtx$ and $\vty$ are fully realized by the Freiman isomorphism
 as $\tI=\phi'(I),\tJ=\phi'(J)$, where $\vtx$ cuts out (in a symmetrical way) for a fixed $\vty\in \C^{\tn}$ an $s\times
 s$ Hermitian matrix $\vB^{\tI}_{\vty}=\vP_{\tI}\vB_{\vty} \vP_{\tI}^*$ (principal submatrix, actually also Toeplitz)
 given by the green elements (here we have re-ordered $I$ such that $\tI$ is ordered) 
 \begin{align}
    \vB_\vty:=  \left(\begin{array}{lllllllll}
         b_0&\cdots &\cred{b_{\ti_0}} &\cdots&\cred{b_{\ti_1}} & \cdots & \cred{b_{\ti_{s-1}}} & \cdots & b_{\tn-1}\\
         \vdots   &\diagdown & \vdots     & & \vdots  &    & \vdots &  &\vdots \\
         \cred{\cc{b_{\ti_0}}} &  \cdots & \cgre{b_{\ti_0-\ti_0}} & \cdots& \cgre{b_{\ti_1-\ti_0}} & \cdots &\cgre{b_{\ti_{s-1}-\ti_0}}&
         \cdots & \cred{b_{\tn-1-\ti_0}} \\
         \vdots   &   &  \vdots       & \diagdown& \vdots  &  & \vdots &  &\vdots \\
         \cred{\cc{b_{\ti_1}}}&\cdots &\cgre{\cc{b_{\ti_1-\ti_0}}} &\cdots& \cgre{{b_{\ti_1-\ti_1}}} & \cdots
         &\cgre{b_{\ti_{s-1}-\ti_1}} & \cdots  & \cred{b_{\tn-\ti_1}} \\
         \vdots   &   & \vdots        & & \vdots  & \diagdown & \vdots & & \vdots \\
         \cred{\cc{b_{\ti_{s-1}}}} & \cdots & \cgre{\cc{b_{\ti_{s-1}-\ti_0}}}  & \cdots& \cgre{\cc{b_{\ti_1-\ti_{s-1}}}}& \cdots
         &\cgre{b_{ \ti_{s-1}-\ti_{s-1}}} &\cdots & \cred{b_{\tn-\ti_{s-1}}} \\
         \vdots   &   &  \vdots     &  & \vdots  &  &\vdots & \diagdown & \vdots \\
         \cc{b_{\tn-1}} & \cdots &\cred{\cc{b_{{\tn-1}-\ti_0}}} & \cdots& \cred{\cc{b_{\tn-1-\ti_1}}} &\cdots &
         \cred{\cc{b_{\tn-1-\ti_{s-1}}}}&  \cdots & b_{0} \\
      \end{array}\right).
 \end{align}
 Minimizing over all $\hvx\in\C^{s}$ we have by \namen{Cauchy}'s
 Interlacing Theorem, see e.g. \cite[Prop.9.19]{BG05a}, for all $s\leq \tn\in\N$
 \begin{align}
    \lam(\vB_\vty^{\tI}) \geq \lam(\vB_\vty)>0 \quad,\quad \vty\in \C^{\tn}, \tI\in [\tn]_s. 
 \end{align}
 Hence, this also holds for $\vty\in \Sigma_{f}^{\tn}$  and we get for our
 problem in \eqref{eq:pb}
 \begin{align}
    \alp^2(s,f)=\inf_{\substack{\vx\in \ell^2_s, \vy\in \ell^2_f\\ \Norm{\vx}=\Norm{\vy}=1}} 
    \Norm{\vx * \vy}
     \geq \min\left\{\min_{{\tI\in[0,\tn-1]_s}}\min_{\substack{\vty\in \Sigma^{\tn}_{f}\\\Norm{\vty}=1}} \lam(\vB^{\tI}_{\vty}) 
     ,\min_{{\tJ\in[0,\tn-1]_f}}\min_{\substack{\vtx\in \Sigma^{\tn}_{s}\\\Norm{\vtx}=1}} \lam(\vB^{\tI}_{\vty})\right\} 
     \geq   \min_{\substack{\va\in\Sigma_{\bet^2(s,f)}^{\tn}\\\Norm{\va}=1}}\lam(\vB_{\va})
   \geq \min_{\substack{\va\in\C^n\\\Norm{\va}=1}}\lam(\vB_{\va}) 
    =:\alp^2_n\label{eq:bys2}.
 \end{align}
 Unfortunately, the combinatoric can only be removed by using the \namen{Cauchy} Interlacing theorem, which obtains only
 a lower bound $\alp_n$ for $\alp(s,f)$. Moreover, the lower bound attained in by the double minimum may still be to
 large: First, $n$ may to large and even if $n$ is the right dimension for the Freiman isomorphism, there are not all
 $\tI\in[0,n-1]_s$ resp.  $\tJ\in[0,n-1]_f$ needed to represent the convolution.
\end{proof}

\section{Conclusion}

There are several applications for this inequality. For example, in \cite{WJ13c} the authors 
have shown a statement for stable phase retrieval from magnitude of $4n-3$ symmetrized Fourier measurements. 
The stability result is independent of the ambient dimension in the regime of $\Omi((s+f-2)\log_2(s+f-2))<\log_2 n$. 
Furthermore, tools from spectral theory of Toeplitz matrices may be used to obtain more precise estimates
for lower bound of the smallest eigenvalue $\alp^2(s,f)$.

\newpage
\section{References}
\vspace{-1cm}
\printbibliography

\end{document}